\documentclass[11pt]{article}

\usepackage[a4paper,margin=30mm]{geometry}
\usepackage{amsmath,amssymb,amsthm,mathtools}
\usepackage{bm}
\usepackage[numbers,sort&compress]{natbib}
\usepackage[hidelinks]{hyperref}

\newcommand{\RR}{\mathbb R}
\newcommand{\Tr}{\operatorname{Tr}}
\newcommand{\Id}{\mathbb I}
\newcommand{\Om}{\Omega}
\newcommand{\Sp}{\operatorname{Sp}}
\newcommand{\Sym}{\operatorname{Sym}}

\newcommand{\PiZ}{\Pi_0}
\newcommand{\PiT}{\Pi_2}

\newtheorem{theorem}{Theorem}[section]
\newtheorem{proposition}[theorem]{Proposition}

\theoremstyle{definition}
\newtheorem{definition}[theorem]{Definition}
\newtheorem{problem}[theorem]{Open problem}
\theoremstyle{remark}
\newtheorem{remark}[theorem]{Remark}

\title{Gaussian Purification Quotients and Fixed Nielsen Penalties}
\author{Christian Kerskens\\
\small Trinity College Institute of Neuroscience, Trinity College Dublin,
Dublin, Ireland}
\date{}

\begin{document}
\maketitle

\begin{abstract}
Information distance and circuit complexity are both obtained by minimizing
lengths, but they minimize over different objects.  We make this distinction
explicit for faithful one-mode Gaussian states.  First, invariant-form
uniqueness implies that no positive-definite quadratic gate cost can be
invariant under the full adjoint action of the noncompact symplectic group;
a positive Cartan majorant necessarily introduces additional reference data.
The Uhlmann purification
quotient realizes the Bures metric, and the radial covariance direction
requires a system--ancilla coupling because system-only Gaussian unitaries
preserve the Williamson eigenvalue.  We then minimize fixed right-invariant
quadratic norms on the minimal two-mode Gaussian gate algebra
\(\mathfrak{sp}(4,\mathbb R)\).  For the unweighted Frobenius norm, the
quotient coefficients for radial and traceless covariance tangents are
\(G_0=[\hbar^2(u-1)]^{-1}\) and
\(G_2=[\hbar^2(3u-1)]^{-1}\), where
\(u=(2\nu/\hbar)^2\).  Their ratio does not equal the Bures ratio. The radial coefficient, however, reproduces the Bures value
exactly at every \(u\); the mismatch is confined to the traceless sector. More
generally, a constant block-diagonal two-weight schedule gives
\(G_0/G_2=1+2(\beta/\alpha)u/(u-1)\); matching Bures throughout the isotropic
family would require the state-dependent relation \(\beta/\alpha=1/u\).
At the Bures--Fisher determinant crossing \(u=\varphi\), pointwise matching
is possible only by inserting \(\beta/\alpha=\varphi^{-1}\).  Thus the Bures
purification quotient is an exact state-geometric cost, but it is neither an
unweighted symplectic gate cost nor a member of this fixed two-weight Nielsen
family.  The existence of a more general fixed positive gate norm realizing
the quotient remains open.
\end{abstract}

\section{Introduction}
\label{sec:introduction}

Nielsen's geometric formulation replaces discrete gate counting by an
optimal-control problem on a group of admissible transformations
\cite{Nielsen2006}.  A positive norm on the generator algebra specifies the
relative cost of elementary operations, and right translation turns that norm
into a state-independent circuit geometry.  The choice of this penalty
schedule is part of the physical definition of the problem.

For Gaussian field-theory states, Nielsen geometries were developed for
pure states in \cite{JeffersonMyers2017,ChapmanEtAl2018}; mixed-state
extensions based on purification, Fisher-information geometry, and the
Uhlmann construction include
\cite{CamargoEtAl2019,DiGiulioTonni2020,Ruan2021}.

Mixed-state information geometry supplies a different minimization.  The
Bures distance is the minimum Fubini--Study distance over purifications
\cite{Uhlmann1976}, and the path-length version and its Gaussian saturation
were analyzed by Ruan \cite{Ruan2020}.  This quotient fixes a state-space
metric.  It does not automatically fix a right-invariant gate norm because
the variance, and hence the Fubini--Study norm, of a fixed generator depends
on the state on which it acts.

The distinction is especially sharp for bosonic Gaussian states.  Their
unitary transformations are symplectic, while mixedness is encoded by the
Williamson eigenvalues \cite{Weedbrook2012}.  Consequently, traceless shape
motion can be generated on the system alone, whereas a radial change of a
one-mode isotropic covariance requires an environment or a purification.
This provides a minimal setting in which a state metric and a gate metric can
be compared by an explicit constrained minimization.

The companion foundation analysis identifies the state-independent
correction to the Gaussian Bures cometric with the Killing form and proves
its invariant uniqueness \cite{KerskensFoundation2026}.  Here we derive its
circuit-geometric consequence: full adjoint symplectic invariance is
incompatible with a positive quadratic Nielsen cost.  We then address the
remaining operational question by constructing the minimal two-mode Gaussian
purification, writing the full \(\mathfrak{sp}(4,\mathbb R)\) lift constraint,
and computing the quotient of two fixed positive norm families.  The general
invariant obstruction and the restricted quotient no-go play complementary
roles: the former explains why extra schedule data are unavoidable, while the
latter tests two concrete positive choices.

Section~\ref{sec:notions} separates the three geometries in play.
Section~\ref{sec:sectors} gives the covariance-sector decomposition and the
system-only obstruction.  Section~\ref{sec:bures-quotient} records the exact
Bures purification quotient.  Sections~\ref{sec:lift} and
\ref{sec:minimization} solve the symplectic generator minimization.
Section~\ref{sec:golden} extracts the golden-point corollary, and
Section~\ref{sec:scope} states the remaining fixed-norm problem.

\section{State distance, purification length, and Nielsen cost}
\label{sec:notions}

Let a centered one-mode Gaussian state have covariance
\begin{equation}
 \Sigma\in\Sym_{++}(2),
 \qquad
 \Sigma+\frac{i\hbar}{2}\Om>0,
 \qquad
 \Om=\begin{pmatrix}0&1\\-1&0\end{pmatrix}.
 \label{eq:domain}
\end{equation}
At an isotropic point we write
\begin{equation}
 \Sigma_\nu=\nu\Id_2,
 \qquad
 \nu>\frac{\hbar}{2},
 \qquad
 u=\left(\frac{2\nu}{\hbar}\right)^2>1.
 \label{eq:isotropic}
\end{equation}

Three length constructions must be kept separate.

\begin{definition}[Bures state length]
For a covariance path \(s\mapsto\Sigma(s)\), the covariance Bures length is
\begin{equation}
 L_{\rm B}[\Sigma]
 =\int_0^1
 \sqrt{g_{{\rm B},\Sigma(s)}
 \bigl(\dot\Sigma(s),\dot\Sigma(s)\bigr)}\,\mathrm ds.
 \label{eq:bures-length}
\end{equation}
It is a length on the reduced state manifold.
\end{definition}

\begin{definition}[Purification quotient]
Let \(|\Psi(s)\rangle\) range over purified paths whose reduced state is
\(\rho(s)\).  The purification length is the infimum of their
Fubini--Study lengths, including minimization over vertical ancillary
motions and purified endpoints.
\end{definition}

\begin{definition}[Right-invariant Nielsen length]
Let \(\mathcal S(s)\) be a path in a fixed Gaussian gate group and
\begin{equation}
 K(s)=\dot{\mathcal S}(s)\mathcal S(s)^{-1}
 \label{eq:right-velocity}
\end{equation}
its right velocity.  A state-independent quadratic schedule is a fixed
positive-definite form \(F(K)^2\) on the gate algebra.  Its circuit length is
\begin{equation}
 L_{\rm N}[\mathcal S]=\int_0^1F(K(s))\,\mathrm ds.
 \label{eq:nielsen-length}
\end{equation}
\end{definition}

The first two constructions coincide after the appropriate purification
minimization.  The third agrees with them only if a fixed positive gate norm,
an admissible generator set, and the reduction map together induce the same
quotient tensor.

The first obstruction is algebraic and does not depend on a particular
purification.

\begin{proposition}[No fully symplectic-invariant quadratic Nielsen cost]
\label{prop:invariant-no-go}
Let \(N\geq1\), and let \(b\) be a symmetric bilinear form on
\(\mathfrak{sp}(2N,\RR)\) satisfying
\begin{equation}
 b(\operatorname{Ad}_S X,\operatorname{Ad}_S Y)=b(X,Y)
 \qquad
 \text{for every }S\in\Sp(2N,\RR).
 \label{eq:ad-invariant-cost}
\end{equation}
Then \(b=c\tau\), where \(\tau(X,Y)=\Tr(XY)\) and \(c\in\RR\).  Consequently,
\(b\) is either degenerate or indefinite and cannot define a
positive-definite Nielsen metric.
\end{proposition}

\begin{proof}
The real symplectic Lie algebra is simple, so every invariant symmetric
bilinear form is proportional to its Killing form, equivalently to the trace
form \(\tau\) \cite{KerskensFoundation2026}.  Under the Cartan decomposition
\begin{equation}
 \mathfrak{sp}(2N,\RR)=\mathfrak k\oplus\mathfrak p,
 \qquad \mathfrak k\cong\mathfrak u(N),
 \label{eq:cartan-split}
\end{equation}
the trace form is negative definite on \(\mathfrak k\) and positive definite
on \(\mathfrak p\).  Every nonzero multiple is therefore indefinite, while
the zero multiple is degenerate.
\end{proof}

A positive replacement can be formed only after an additional choice.  For
the standard Cartan involution \(\theta(X)=-X^{\mathsf T}\), define
\begin{equation}
 \langle X,Y\rangle_\theta
 =-\tau(X,\theta Y)=\Tr(XY^{\mathsf T}).
 \label{eq:cartan-majorant}
\end{equation}
This Cartan majorant is positive definite, but the choice of \(\theta\), or
equivalently of a reference Euclidean complex structure, is part of the gate
schedule.  The form is invariant under the associated maximal compact
subgroup rather than under the full adjoint action of \(\Sp(2N,\RR)\).
Proposition~\ref{prop:invariant-no-go} does not prohibit positive
right-invariant Nielsen geometries: right invariance does not require adjoint
invariance.  It shows instead that every positive schedule necessarily
contains structure not fixed by full symplectic invariance alone.

\section{Covariance sectors and the system-only obstruction}
\label{sec:sectors}

Use the trace-orthonormal basis
\begin{equation}
 E_0=\frac{\Id_2}{\sqrt2},\qquad
 E_+=\frac1{\sqrt2}\begin{pmatrix}1&0\\0&-1\end{pmatrix},\qquad
 E_\times=\frac1{\sqrt2}\begin{pmatrix}0&1\\1&0\end{pmatrix}.
 \label{eq:basis}
\end{equation}
Then
\begin{equation}
 \Sym(2)=\mathbb RE_0\oplus
 \operatorname{span}\{E_+,E_\times\}.
 \label{eq:sector-split}
\end{equation}
The scalar sector changes the covariance scale and the two-dimensional
traceless sector changes its shape.

An infinitesimal system-only Gaussian unitary acts by
\begin{equation}
 \delta_Y\Sigma=Y\Sigma+\Sigma Y^{\mathsf T},
 \qquad Y\in\mathfrak{sp}(2,\mathbb R).
 \label{eq:orbit-tangent}
\end{equation}

\begin{proposition}[System-only unitary obstruction]
\label{prop:system-obstruction}
At \(\Sigma_\nu=\nu\Id_2\), the tangent space of the symplectic orbit is
exactly the traceless sector:
\begin{equation}
 T_{\Sigma_\nu}\bigl(\Sp(2,\mathbb R)\cdot\Sigma_\nu\bigr)
 =\operatorname{span}\{E_+,E_\times\}.
 \label{eq:orbit-sector}
\end{equation}
The radial direction \(E_0\) cannot be generated by a system-only Gaussian
unitary.
\end{proposition}

\begin{proof}
At the isotropic point,
\(\delta_Y\Sigma_\nu=\nu(Y+Y^{\mathsf T})\).  Every element of
\(\mathfrak{sp}(2,\mathbb R)\) is traceless, so this symmetric tangent is
traceless.  Conversely, each symmetric traceless \(2\times2\) matrix \(X\)
is symplectic and is generated by \(Y=X/(2\nu)\).  Radial motion changes
\(\det\Sigma\), equivalently the Williamson eigenvalue, which is invariant
under symplectic congruence.
\end{proof}

The obstruction identifies which controls must be present in an enlarged
circuit, but it does not assign their costs.  In particular, it does not turn
the indefinite Schur fiber of the companion foundation analysis into a
positive ancilla metric.

\section{The exact Bures purification quotient}
\label{sec:bures-quotient}

Uhlmann's theorem identifies the Bures distance with the quotient of
the Fubini--Study distance over purifications \cite{Uhlmann1976}.  In
path-length form, the squared Bures speed equals the minimum squared
Fubini--Study speed among purified lifts, attained by the horizontal
lift \cite{Uhlmann1976,DittmannUhlmann1999}; the equivalence of
mixed-state purification complexity with this quotient is established
in general by Ruan \cite{Ruan2021}.  For Gaussian families the quotient
closes within the Gaussian category:

\begin{proposition}[Gaussian saturation of the Uhlmann quotient]
\label{prop:gaussian-saturation}
For a faithful one-mode Gaussian covariance path, the purification
quotient is attained within the family of pure two-mode Gaussian
purifications: the Fubini--Study speed of the horizontal Gaussian lift
equals the Bures speed of the reduced path.  No non-Gaussian ancilla is
needed for this state-geometric minimization.
\end{proposition}

\begin{proof}[Proof sketch]
The symmetric logarithmic derivative of a faithful centered Gaussian
state is quadratic in the canonical operators
\cite{Monras2013,Banchi2015}.  The horizontal-lift condition
\(\dot W=\tfrac12LW\) on amplitudes therefore has a quadratic
generator, and quadratic generators close under the flow, so the
horizontal lift of a Gaussian path remains Gaussian and its
Fubini--Study speed attains the Uhlmann minimum.  The horizontal
generator is Hermitian and norm-preserving but is not itself a circuit
unitary; a unitary circuit generator inducing the same reduced tangent
is reconstructed separately in \(\mathfrak{sp}(4,\RR)\), up to an
overall phase, in Section~\ref{sec:lift}.
\end{proof}

At \(\Sigma_\nu\), let \(\PiZ\) and \(\PiT\) be the trace-orthogonal
projectors onto the scalar and traceless sectors.  Relative to the classical
covariance Fisher tensor, the Bures cometric has sector eigenvalues
\begin{equation}
 r_0(u)=\frac{u-1}{u},
 \qquad
 r_2(u)=\frac{u+1}{u}.
 \label{eq:bures-cometric-ratios}
\end{equation}
Hence the corresponding relative metric is
\begin{equation}
 G_{\rm B}^{\rm rel}(u)
 =\frac{u}{u-1}\PiZ+\frac{u}{u+1}\PiT,
 \label{eq:bures-relative-metric}
\end{equation}
and its radial-to-traceless squared-cost ratio is
\begin{equation}
 R_{\rm B}(u)=\frac{u+1}{u-1}.
 \label{eq:bures-ratio}
\end{equation}

Equation~\eqref{eq:bures-relative-metric} is an exact local state-space
penalty.  It is state dependent through \(u\).  To decide whether it is also
a Nielsen penalty, one must minimize a fixed norm on a fixed enlarged gate
algebra and compare the resulting quotient.

\section{Minimal two-mode Gaussian lift}
\label{sec:lift}

A centered Gaussian channel acts on covariances as
\begin{equation}
 \Sigma\longmapsto M\Sigma M^{\mathsf T}+N,
 \label{eq:gaussian-channel}
\end{equation}
with complete positivity condition
\begin{equation}
 N+\frac{i\hbar}{2}
 \left(\Om-M\Om M^{\mathsf T}\right)\geq0.
 \label{eq:cp-condition}
\end{equation}
Such reduced dynamics can be realized by a global Gaussian unitary on the
system and a Gaussian environment followed by a partial trace
\cite{Weedbrook2012}.  This supplies the natural control architecture for
radial motion, but a channel dilation alone still does not choose a gate norm.

\begin{definition}[Extended Gaussian Nielsen problem]
Fix a reference purification, a target reduced covariance, a gate algebra
split into system, ancilla, and interaction controls, a positive
state-independent right-invariant norm on that algebra, and an endpoint
equivalence relation identifying purified endpoints with the same reduced
target.  The reduced circuit cost is the infimum of the total Nielsen length
over all admissible circuits and equivalent purified endpoints.
\end{definition}

The local calculation below fixes the minimal ancillary dimension and a
reduced tangent rather than a global endpoint.  It is the infinitesimal
quotient problem associated with this definition.

The canonical one-ancilla purification of \(\Sigma_\nu\) has covariance
\begin{equation}
 \Gamma_\nu=
 \begin{pmatrix}
  \nu\Id_2&\kappa_\nu Z\\
  \kappa_\nu Z&\nu\Id_2
 \end{pmatrix},
 \qquad
 \kappa_\nu=\sqrt{\nu^2-\frac{\hbar^2}{4}},
 \qquad
 Z=\begin{pmatrix}1&0\\0&-1\end{pmatrix}.
 \label{eq:purification}
\end{equation}
With \(\Om_4=\Om\oplus\Om\), it obeys
\begin{equation}
 \Gamma_\nu\Om_4\Gamma_\nu=\frac{\hbar^2}{4}\Om_4,
 \label{eq:purity}
\end{equation}
so it is pure and has system marginal \(\Sigma_\nu\).

A generator \(K\in\mathfrak{sp}(4,\mathbb R)\) has block form
\begin{equation}
 K=\begin{pmatrix}
 A&B\\
 \Om B^{\mathsf T}\Om&D
 \end{pmatrix},
 \qquad
 A,D\in\mathfrak{sp}(2,\mathbb R).
 \label{eq:sp4-block}
\end{equation}
The global covariance velocity
\(\dot\Gamma=K\Gamma_\nu+\Gamma_\nu K^{\mathsf T}\) induces the prescribed
system tangent \(X\) precisely when
\begin{equation}
 X=\nu(A+A^{\mathsf T})
 +\kappa_\nu(BZ+ZB^{\mathsf T}).
 \label{eq:constraint}
\end{equation}
The local ancilla block \(D\) is vertical for this reduced constraint.  The
trace of the first term in Eq.~\eqref{eq:constraint} vanishes, so the
entangling block \(B\) must carry every radial target, as anticipated by
Proposition~\ref{prop:system-obstruction}.

\section{Fixed right-invariant norm quotients}
\label{sec:minimization}

We now extend a positive quadratic norm at the identity to a right-invariant
metric on \(\Sp(4,\mathbb R)\) and minimize it subject to
Eq.~\eqref{eq:constraint}.  This is a gate-level quotient, distinct from the
Fubini--Study purification quotient of Section~\ref{sec:bures-quotient}.

\subsection{The standard unweighted Nielsen schedule (Frobenius norm)}

Consider
\begin{equation}
 F_{\rm Frob}(K)^2=\frac14\Tr(K^{\mathsf T}K).
 \label{eq:frobenius}
\end{equation}
This is a positive state-independent norm at the identity and defines a
right-invariant metric by translation.  Its quadrature basis is part of the
schedule; it is not invariant under the full adjoint symplectic action.
Because \(D\) does not enter the constraint and is orthogonal to the other
blocks, the optimum has \(D=0\).  Moreover,
\(\|\Om B^{\mathsf T}\Om\|_{\rm F}=\|B\|_{\rm F}\), so
\begin{equation}
 F_{\rm Frob}(A,B)^2
 =\frac14\left[
 \Tr(A^{\mathsf T}A)+2\Tr(B^{\mathsf T}B)
 \right].
 \label{eq:frobenius-reduced}
\end{equation}

Rotational symmetry makes the two traceless basis directions equivalent.
For the diagonal representatives \(X_0=\Id_2\) and \(X_2=Z\), off-diagonal
generator components form a decoupled nonnegative quadratic block and vanish
at the minimum.  These two targets have equal covariance Fisher norm at
\(\Sigma_\nu\), so the ratio of their quotient coefficients can be compared
directly with Eq.~\eqref{eq:bures-ratio}.  Explicitly, write
\begin{equation}
 A=\begin{pmatrix}a&b\\c&-a\end{pmatrix},
 \qquad
 B=\begin{pmatrix}d&e\\f&g\end{pmatrix}.
 \label{eq:general-generators}
\end{equation}
For a diagonal target, the only off-diagonal constraint is
\begin{equation}
 \nu(b+c)+\kappa_\nu(f-e)=0.
 \label{eq:offdiagonal-constraint}
\end{equation}
The variables \((b,c,e,f)\) do not enter either diagonal constraint, while
their contribution to Eq.~\eqref{eq:frobenius-reduced} is a positive sum of
squares.  Their unique minimum is zero.  It is therefore sufficient to set
\begin{equation}
 A=\operatorname{diag}(a,-a),
 \qquad B=\operatorname{diag}(b_{11},b_{22}).
 \label{eq:diagonal-ansatz}
\end{equation}

For \(X_0=\Id_2\), the constraint gives
\begin{equation}
 b_{11}=\frac{1-2\nu a}{2\kappa_\nu},
 \qquad
 b_{22}=\frac{-1-2\nu a}{2\kappa_\nu},
 \label{eq:radial-b}
\end{equation}
and hence
\begin{equation}
 E_0(a)=\frac12a^2+
 \frac{1+4\nu^2a^2}{4\kappa_\nu^2}.
 \label{eq:radial-energy}
\end{equation}
The optimum is \(a=0\); radial motion is carried entirely by the interaction
block.  Its quotient coefficient is
\begin{equation}
 G_0^{\rm Frob}=\frac{1}{4\kappa_\nu^2}
 =\frac{1}{\hbar^2(u-1)}.
 \label{eq:frobenius-radial}
\end{equation}

For \(X_2=Z\), both diagonal entries of \(B\) equal
\begin{equation}
 b_{11}=b_{22}=\frac{1-2\nu a}{2\kappa_\nu},
 \label{eq:shear-b}
\end{equation}
and
\begin{equation}
 E_2(a)=\frac12a^2+
 \frac{(1-2\nu a)^2}{4\kappa_\nu^2}.
 \label{eq:shear-energy}
\end{equation}
The minimizer and coefficient are
\begin{equation}
 a_\star=\frac{\nu}{\kappa_\nu^2+2\nu^2},
 \qquad
 G_2^{\rm Frob}
 =\frac{1}{4(\kappa_\nu^2+2\nu^2)}
 =\frac{1}{\hbar^2(3u-1)}.
 \label{eq:frobenius-shear}
\end{equation}
Although the shear is system-only generable, its quotient-optimal lift uses
ancilla assistance because the reduction permits \(B\neq0\).

\begin{proposition}[Failure of the unweighted schedule]
The Frobenius quotient has ratio
\begin{equation}
 R_{\rm Frob}(u)=\frac{G_0^{\rm Frob}}{G_2^{\rm Frob}}
 =\frac{3u-1}{u-1},
 \label{eq:frobenius-ratio}
\end{equation}
which differs from the Bures ratio \(R_{\rm B}(u)\) for every \(u>1\).
\end{proposition}

\begin{proof}
Equations~\eqref{eq:frobenius-radial} and
\eqref{eq:frobenius-shear} give the first equality.  Equality with
Eq.~\eqref{eq:bures-ratio} would require \(3u-1=u+1\), hence \(u=1\), which
lies at the excluded faithful boundary.
\end{proof}

\begin{proposition}[Exact radial agreement]
\label{prop:radial-agreement}
At every \(u>1\), the unweighted Frobenius quotient reproduces the
Bures radial coefficient identically,
\begin{equation}
 G_0^{\rm Frob}
 =g_{{\rm B},\Sigma_\nu}(X_0,X_0)
 =\frac{1}{\hbar^2(u-1)},
 \label{eq:radial-coincidence}
\end{equation}
and its radial minimizer \(A=0\), \(B=Z/(2\kappa_\nu)\) generates the
tangent of the canonical purification family:
\begin{equation}
 K_\star\Gamma_\nu+\Gamma_\nu K_\star^{\mathsf T}
 =\frac{\mathrm d}{\mathrm d\nu}\Gamma_\nu .
 \label{eq:minimizer-is-lift}
\end{equation}
The discrepancy of the unweighted schedule is therefore confined to the
traceless sector, where
\begin{equation}
 G_2^{\rm Frob}=\frac{1}{\hbar^2(3u-1)}
 <\frac{1}{\hbar^2(u+1)}
 =g_{{\rm B},\Sigma_\nu}(X_2,X_2).
 \label{eq:shear-undercost}
\end{equation}
\end{proposition}

\begin{proof}
The radial SLD gain at \(\Sigma_\nu\) with target \(X_0=\Id_2\) is
\(\mathfrak G=\Id_2/\kappa_\nu^2\), since
\(\Om\Id_2\Om=-\Id_2\); hence
\(g_{\rm B}(X_0,X_0)=\tfrac18\Tr(\mathfrak GX_0)
=1/(4\kappa_\nu^2)=[\hbar^2(u-1)]^{-1}\), which is
Eq.~\eqref{eq:frobenius-radial}.  For
Eq.~\eqref{eq:minimizer-is-lift}, with \(B=bZ\) one has
\(\Om B^{\mathsf T}\Om=bZ\), and block multiplication against
Eq.~\eqref{eq:purification} gives system and ancilla blocks
\(2b\kappa_\nu\Id_2\) and off-diagonal blocks \(2b\nu Z\); at
\(b=1/(2\kappa_\nu)\) these equal \(\Id_2\) and
\((\nu/\kappa_\nu)Z=(\mathrm d\kappa_\nu/\mathrm d\nu)Z\),
matching \(\partial_\nu\Gamma_\nu\) blockwise.  The traceless gain
\(\mathfrak G=Z/(\nu^2+\hbar^2/4)\) follows from \(\Om Z\Om=Z\), giving
\(g_{\rm B}(X_2,X_2)=[\hbar^2(u+1)]^{-1}\), and
\(3u-1>u+1\) for \(u>1\).
\end{proof}

\begin{remark}[What the radial agreement fixes]
\label{rem:radial-agreement}
Three consequences.  First, along radial motion the gate quotient, the
Uhlmann quotient, and the Bures length coincide: the Frobenius-optimal
circuit is purely entangling and tracks the same lift whose
Fubini--Study speed saturates the Uhlmann bound
(Proposition~\ref{prop:gaussian-saturation}), and the unweighted
schedule prices it at exactly its Bures cost.  Second, the overall
scale \(1/4\) in Eq.~\eqref{eq:frobenius} is thereby no longer a
convention: rescaling the norm preserves the ratio
\(R_{\rm Frob}\) but destroys Eq.~\eqref{eq:radial-coincidence}, so
radial agreement fixes the schedule's absolute normalization.  Third,
the sector pattern is sharp: the quotient and the state geometry agree
exactly on the sector where the reduction admits only entangling
controls, and disagree exactly where system-only and ancilla-assisted
routes compete --- there the circuit exploits the cheaper mixture,
Eq.~\eqref{eq:shear-undercost}, undercutting even the system-only cost
\([2\hbar^2u]^{-1}\).  The failure of the unweighted schedule is thus a
failure of shear pricing alone.
\end{remark}
\subsection{The two-weight Nielsen penalty schedule}

To expose the remaining normalization freedom, consider
\begin{equation}
 F_{\alpha,\beta,\gamma}(K)^2
 =\frac14\left[
 \alpha\Tr(A^{\mathsf T}A)
 +2\beta\Tr(B^{\mathsf T}B)
 +\gamma\Tr(D^{\mathsf T}D)
 \right],
 \qquad \alpha,\beta,\gamma>0.
 \label{eq:weighted-norm}
\end{equation}
Again \(D=0\) at the quotient minimum.  The same constrained minimization
replaces Eqs.~\eqref{eq:radial-energy} and \eqref{eq:shear-energy} by
\begin{align}
 E_0^{(\alpha,\beta)}(a)
 &=\frac{\alpha}{2}a^2
 +\frac{\beta(1+4\nu^2a^2)}{4\kappa_\nu^2},
 \label{eq:weighted-radial-energy}\\
 E_2^{(\alpha,\beta)}(a)
 &=\frac{\alpha}{2}a^2
 +\frac{\beta(1-2\nu a)^2}{4\kappa_\nu^2}.
 \label{eq:weighted-shear-energy}
\end{align}
The radial minimizer remains \(a=0\), while the shear minimizer is
\begin{equation}
 a_\star^{(\alpha,\beta)}
 =\frac{\beta\nu}{\alpha\kappa_\nu^2+2\beta\nu^2}.
 \label{eq:weighted-minimizer}
\end{equation}
Substitution gives
\begin{align}
 G_0^{(\alpha,\beta)}
 &=\frac{\beta}{\hbar^2(u-1)},
 \label{eq:weighted-radial}\\
 G_2^{(\alpha,\beta)}
 &=\frac{\alpha\beta}
 {\hbar^2\left[\alpha(u-1)+2\beta u\right]}.
 \label{eq:weighted-shear}
\end{align}
Therefore
\begin{equation}
 R_{\rm N}(u)
 :=\frac{G_0^{(\alpha,\beta)}}{G_2^{(\alpha,\beta)}}
 =1+2\frac{\beta}{\alpha}\frac{u}{u-1}.
 \label{eq:weighted-ratio}
\end{equation}

\begin{theorem}[Restricted fixed-schedule obstruction]
\label{thm:two-weight-no-go}
No state-independent choice of positive constants \(\alpha,\beta,\gamma\)
in Eq.~\eqref{eq:weighted-norm} reproduces the Bures radial-to-traceless
squared-cost ratio throughout the faithful isotropic family.  Pointwise
agreement at a specified \(u\) requires
\begin{equation}
 \frac{\beta}{\alpha}=\frac1u.
 \label{eq:required-weight}
\end{equation}
\end{theorem}

\begin{proof}
Equating Eqs.~\eqref{eq:weighted-ratio} and \eqref{eq:bures-ratio} gives
\begin{equation}
 1+2\frac{\beta}{\alpha}\frac{u}{u-1}
 =1+\frac{2}{u-1},
\end{equation}
and hence Eq.~\eqref{eq:required-weight}.  Since \(u\) varies along the
family, no constant ratio \(\beta/\alpha\) satisfies the equality on an open
interval.
\end{proof}

Theorem~\ref{thm:two-weight-no-go} is deliberately restricted.  It excludes
the unweighted schedule and the natural block-diagonal two-weight family, but
not every positive inner product on \(\mathfrak{sp}(4,\mathbb R)\).  General
fixed cross terms compatible with a chosen compact symmetry have not yet been
classified.

\section{The golden crossing as a pointwise corollary}
\label{sec:golden}

The relative Bures and covariance Fisher cometric determinants coincide when
\begin{equation}
 \frac{u-1}{u}\left(\frac{u+1}{u}\right)^2=1.
 \label{eq:det-crossing}
\end{equation}
The unique solution in \(u>1\) is
\(u=\varphi=(1+\sqrt5)/2\).  At this point,
\begin{equation}
 Q_{\rm gold}=\varphi^{-2}\PiZ+\varphi\PiT,
 \qquad
 G_{\rm gold}=Q_{\rm gold}^{-1}
 =\varphi^2\PiZ+\varphi^{-1}\PiT,
 \label{eq:gold-operators}
\end{equation}
so the Bures radial-to-traceless squared-cost ratio is
\begin{equation}
 R_{\rm B}(\varphi)=\varphi^3.
 \label{eq:gold-bures-ratio}
\end{equation}

By contrast, the unweighted quotient gives
\begin{equation}
 R_{\rm Frob}(\varphi)=2\varphi+3\neq\varphi^3.
 \label{eq:gold-frobenius-ratio}
\end{equation}
The two-weight schedule agrees at this single point only if
\begin{equation}
 \frac{\beta}{\alpha}=\varphi^{-1}.
 \label{eq:gold-weight}
\end{equation}
Thus the golden ratio is an exact state-geometric specialization, not a
parameter-free consequence of an unweighted Gaussian gate algebra.  A
pointwise Nielsen match is possible, but only after the same scale has been
inserted into the gate schedule.

In the holographic-complexity context, where identifications between
circuit cost and geometric quantities are sought
\cite{Susskind2016,Brown2016}, the crossing has a limited but precise
significance.  A fixed two-weight schedule can agree with the Bures
penalty at exactly one point of the isotropic family, and only by
importing the state-geometric value \(\beta/\alpha=\varphi^{-1}\) into
the hardware.  The golden point therefore measures the mismatch between
the two minimizations; it does not resolve it.
\section{Scope and the remaining fixed-norm problem}
\label{sec:scope}

The calculations above establish four statements with different logical
status:
\begin{enumerate}
 \item The Bures metric is exactly the Fubini--Study purification quotient;
 for the one-mode Gaussian family the optimum is Gaussian.
 \item Radial covariance motion is transverse to the system-only symplectic
 orbit and therefore requires an enlarged or nonunitary realization.
 \item The unweighted right-invariant Frobenius quotient does not reproduce
 the Bures sector ratio.
 \item No constant schedule in the block-diagonal two-weight family
 Eq.~\eqref{eq:weighted-norm} reproduces that ratio on the full isotropic
 family.
\end{enumerate}

These local statements do not yet determine a global preparation complexity.
That problem additionally requires a reference state, a target, an ancillary
dimension rule, allowed endpoints, and minimization over complete paths.  Nor
does the algebraic Schur lift of the Bures cometric determine a physical
channel cost: its fiber is indefinite and has a different purpose.

\begin{problem}[Fixed positive realization]
Classify the positive state-independent inner products on a fixed enlarged
Gaussian gate algebra, including symmetry-allowed cross terms, whose quotient
under the reduced-covariance map equals the Bures metric on a nontrivial open
family of mixed Gaussian states.  If no such inner product exists under a
natural compact-invariance and locality class, determine the largest class
for which a no-go theorem can be proved.
\end{problem}

\begin{remark}[A four-dimensional equivariance constraint beyond
Gaussian quotients]
\label{rem:d4-equivariance}
The sector split of Eq.~\eqref{eq:sector-split} carries weights \(0\)
and \(\pm2\) under the internal phase-space rotation of a single mode.
In spacetime dimension four, and only there, the little group
\(SO(d-2)\) of a causal-horizon cut is likewise a \(U(1)\), and the
trace and transverse--traceless graviton polarizations carry the same
weights.  Any macroscopic lift of the sector costs computed here is
therefore constrained by \(U(1)\) equivariance to enter as an angular
integral over cut directions, and a necessary --- not sufficient ---
consistency target for such a lift is the gauge-fixed Fierz--Pauli
coefficient ratio \(-(d-1)/d\).  The group isomorphism does not by
itself supply the intertwiner, and the construction does not extend
past \(d=4\), where the little group is nonabelian.  We record this as
a constraint on any future identification, not as a result of the
present quotient analysis.
\end{remark}

This problem is the natural location for further Nielsen-complexity
work, particularly in view of holographic proposals relating boundary
complexity to bulk geometry \cite{Susskind2016,Brown2016}.
Proposition~\ref{prop:invariant-no-go} shows that a bi-invariant
quadratic cost --- right-invariant and invariant under the full adjoint
action --- is necessarily indefinite or degenerate for continuous
variables, so every positive schedule contains reference structure not
fixed by symplectic symmetry.
Theorem~\ref{thm:two-weight-no-go} shows that within the block-diagonal
two-weight class this reference structure cannot be chosen
state-independently so as to reproduce the Bures penalty.  A dictionary
identifying bulk geometry with a fixed boundary schedule must therefore
specify how the reference structure is chosen and over which norm class
the identification is claimed; whether some fixed norm with
symmetry-allowed cross terms evades the obstruction is precisely the
open problem above.

\section{Conclusion}

Gaussian Bures geometry supplies a rigorous purification cost, while
Nielsen complexity requires a fixed positive norm on generators.  At
the algebraic level the two cannot be identified by symmetry alone: the
unique adjoint-invariant quadratic form on the symplectic algebra is
indefinite, so every positive schedule contains reference structure
beyond symplectic invariance
(Proposition~\ref{prop:invariant-no-go}).

For the minimal two-mode lift the comparison is exact and
sector-resolved.  The unweighted Frobenius quotient reproduces the
Bures radial coefficient identically --- its optimal circuit generates
the canonical purification family, so the gate quotient, the Uhlmann
quotient, and the Bures length coincide on radial motion
(Propositions~\ref{prop:gaussian-saturation}
and~\ref{prop:radial-agreement}) --- while it underprices the traceless
sector at every faithful \(u\).  Within the constant two-weight family
the discrepancy cannot be repaired: agreement throughout the isotropic
family would force the state-dependent schedule
\(\beta/\alpha=1/u\), and pointwise agreement is available only at the
determinant crossing \(u=\varphi\), at the price of inserting
\(\beta/\alpha=\varphi^{-1}\) into the gate schedule
(Theorem~\ref{thm:two-weight-no-go}).

These are restricted no-go statements, not a general obstruction.  The
pattern they expose --- agreement exactly where the reduction admits
only entangling controls, failure exactly where system-only and
ancilla-assisted routes compete --- suggests that the Bures quotient is
not a fixed-schedule cost in any natural class, but the classification
of positive norms with symmetry-allowed cross terms remains open
(Section~\ref{sec:scope}).  Whether some such norm realizes the
quotient, or a no-go extends to the full class, determines how sharply
holographic identifications of state geometry with circuit cost must be
qualified.  The behavior of the quotient under open-system dynamics and
its role at thermodynamic junctions are taken up in companion work.
\small
\section*{Acknowledgments}
The author acknowledges the use of AI assistants for structural brainstorming, language refinement, and \LaTeX{} editing during preparation of the manuscript. The author bears full responsibility for all arguments, equations, and citations.

\bibliographystyle{unsrtnat}
\bibliography{complexity_references}

\end{document}